\documentclass[preprint]{sigplanconf}

\usepackage{color}
\usepackage{amsmath}
\usepackage{amssymb}
\usepackage{graphicx}
\usepackage{subfigure}
\usepackage{algorithm2e}
\usepackage{dsfont}
\usepackage{amsfonts}
\usepackage{amsthm}

\def\>{\ensuremath{\rangle}}
\def\<{\ensuremath{\langle}}
\def\ra{\ensuremath{\rightarrow}}

\def\la{\ensuremath{\longleftarrow}}
\newcommand {\supp } {{\rm supp}}
\newcommand {\spans } {{\rm span}}
\newcommand {\E } {{\mathcal{E}}}
\newcommand {\F } {{\mathcal{F}}}
\newcommand {\R } {{\mathcal{R}}}
\newcommand{\hs}{\mathcal{H}}
\newcommand {\M} {{\mathcal{M}}}
\newcommand {\s} {{\mathfrak{S}}}
\newcommand {\tr} {{\mathrm{tr}}}
\newcommand {\rd} {{\mathrm{rd}}}
\newcommand{\mem}{{\mathbf{M}}}

\newtheorem{thm}{Theorem}[section]

\newtheorem{lem}{Lemma}[section]
\newtheorem{defi}{Definition}[section]
\newtheorem{prop}{Proposition}[section]
\newtheorem{exam}{Example}[section]

\newtheorem{prob}{Problem}[section]

\begin{document}

\setlength{\pdfpageheight}{\paperheight}
\setlength{\pdfpagewidth}{\paperwidth}

\conferenceinfo{CONF 'yy}{Month d--d, 20yy, City, ST, Country}
\copyrightyear{20yy}
\copyrightdata{978-1-nnnn-nnnn-n/yy/mm}
\doi{nnnnnnn.nnnnnnn}

\title{Reachability Analysis of Quantum Markov Decision Processes}

\authorinfo{Shenggang Ying and Mingsheng Ying}
           {Tsinghua University, China and University of Technology, Sydney, Australia}
           {yingshenggang@gmail.com, Mingsheng.Ying@uts.edu.au, yingmsh@tsinghua.edu.cn}

\maketitle

\begin{abstract}
We introduce the notion of quantum Markov decision process (qMDP) as a semantic model of nondeterministic and concurrent quantum programs. It is shown by examples that qMDPs can be used in analysis of quantum algorithms and protocols. We study various reachability problems of qMDPs both for the finite-horizon and for the infinite-horizon. The (un)decidability and complexity of these problems are settled, or their relationships with certain long-standing open problems are clarified. We also develop an algorithm for finding optimal scheduler that attains the supremum reachability probability.

\end{abstract}




\section{Introduction}

As a generalisation of Markov chains, Markov decision processes (MDPs) stemmed from operations research in 1950's. Now they have been successfully applied in various areas such as economics and finance, manufacturing, control theory, robotics, artificial intelligence and machine learning. Also, effective analysis and resolution techniques for MDPs like linear programming have been developed in the last six decades. Since Vardi \cite{Var85} proposed to adopt MDPs as a model of concurrent probabilistic programs, MDPs have been widely used in analysis and verification of randomised algorithms and probabilistic programs (see, for instance, \cite{Mon05}) as well as model checking of probabilistic computing systems \cite{BK08}.

In this paper we introduce the notion of quantum Markov decision process (qMDP) as a model of nondeterministic and concurrent quantum programs. Research on quantum programming has been intensively conducted in the last 18 years since Knill \cite{Kn96} introduced the Quantum Random Access Machine model for quantum computing and proposed a set of conventions for writing quantum pseudocode. The research includes design of quantum programming languages, e.g. QCL \cite{O03}, qGCL \cite{SZ00}, QPL \cite{Se04} and Quipper \cite{Gre13}, semantic models of quantum programs \cite{DP06}, and verification of quantum programs \cite{Y11} (we refer the reader to \cite{Ga06} for basic ideas of quantum programming and an excellent survey on the early works in this area). In particular, quantum Markov chains were defined in \cite{Y13, YFYY} for modelling sequential quantum programs. This paper extends quantum Markov chains considered in \cite{Y13, YFYY} to qMDPs so that we can model nondeterministic and concurrent quantum programs \cite{Zu04, YY12}.

A classical MDP consists of a set $S$ of states and a set $Act$ of actions. Each action $\alpha\in Act$ is modelled by a probabilistic transition function $P_\alpha:S\times S\rightarrow [0,1]$ with $P_\alpha(s,s^\prime)$ being the probability that the system moves from state $s$ to $s^\prime$ after action $\alpha$. A MDP allows not only probabilistic choice between the system states as a result of performing an action but also a nondeterministic choice between actions: there may be more than one action enabled on entering a state $s$. Thus, the notion of scheduler was introduced to resolve the nondeterministic choice between the enabled actions. A scheduler selects the next action according to the previous and current states of the system.
A qMDP is defined as quantum generalisation of MDP with the set $S$ of states replaced by a Hilbert space $\hs$ which always serves as the state space of a quantum system in physics. Now each action $\alpha\in Act$ is described by a super-operator $\E_\alpha$ in $\hs$. Super-operators were recognised by physicists as the most general mathematical formalism of physically realisable operations in quantum mechanics \cite{NC00}. They were also adopted as denotational semantics of quantum programs by Selinger \cite{Se04} and D'Hont and Panangaden \cite{DP06} in their pioneering works on quantum programming.

A major conceptual difference between classical MDPs and qMDPs comes from the notion of scheduler. The information used by a scheduler in a MDP to select the next action is the state of the system. In the quantum case, however, we choose to introduce a series of measurements at the middle of the evolution of a qMDP and to define a scheduler as a function that selects the next action according to the outcomes of these measurements.

This paper focuses on the aspect of qMDPs more related to program analysis and verification, namely reachability analysis. As in the case of classical MDPs, we consider the reachability probability of a subspace $B$ of the state Hilbert space of a qMDP with a fixed scheduler and the supremum reachability probability of $B$ over all schedulers. Although the definition of reachability probabilities in qMDPs looks similar to that of classical MDPs, their behaviours are very different; for example, a MDP has an optimal scheduler that can achieve the supremum reachability probability for all initial states. But it is not the case in a qMDP even for a given initial state. It is also interesting to observe the difference between the behaviour of qMDPs and that of quantum Markov chains. It was proved in \cite{YFYY} that a quantum Markov chain eventually reaches a subspace $B$ for any initial state if the ortho-complement $B^\perp$ of $B$ in the state Hilbert space $\hs$ contains no bottom strongly connected components (BSCCs). The corresponding notion of BSCC in a qMDP is invariant subspace. However, it is possible that in a qMDP $B^\perp$ contains no invariant subspaces but for some schedulers, $B$ is reached by a probability smaller than $1$.

As indicated in Subsection \ref{Ex-QMDP}, some problems in the analysis of quantum algorithms can be properly formulated as the reachability problem of qMDPs. We believe that it will be inevitable to develop effective techniques for reachability analysis of qMDPs with applications in quantum program analysis and verification as quantum algorithm and program design become more and more sophisticated.

The aspects of qMDPs more related to decision making and machine learning are left for future research. In the last few years, it has been found that probabilistic programming is very useful in machine learning for describing probabilistic distributions and Bayesian inference (see, for instance, \cite{Gor14}). On the other hand, it was realised recently that a major application area of quantum computing might be machine learning and big data analytics. We expect that qMDPs will serve as a bridge between the researches on quantum programming and quantum machine learning.

\textbf{\textit{Contribution of the paper}}: This paper studies (un)decidability and complexity of reachability analysis for qMDPs. In the case of finite-horizon, it is proved that both quantitative reachability and qualitative reachability of qMDPs are undecidable. In the case of infinite-horizon, we show that it is EXPTIME-hard to decide whether the supremum reachability probability of a qMDP is $1$, and if it is smaller than $1$, then the supremum reachability probability is uncomputable. It is further proved that a qMDP has an optimal scheduler for reaching an invariant  subspace of its state Hilbert space if and only if the ortho-complement of the target subspace contains no invariant subspaces. This result enables us to develop an algorithm for finding an optimal scheduler. We also consider the problem whether a qMDP always reach an invariant subspace with probability 1, no matter what the scheduler is. A connection between this problem and a long-standing open problem -  the joint spectral radius problem \cite{Gurvits95, TB97, CM05} - is observed.

\textbf{\textit{Related work}}: Before this paper, a very interesting paper by Barry, Barry and Aaronson \cite{Barry14} was recently posted at http://arxiv .org/abs/1406.2858 where the notion of quantum partially observable Markov decision process was introduced. It was proved in \cite{Barry14} that reachability of a goal state is undecidable in the quantum case but decidable in the classical case. The undecidability in the quantum case is similar to our Theorem \ref{thm:undebabilitySpaceFiniteAnyInv}, but they are not the same since we consider reachability of invariant subspaces rather than a single state. Other results in \cite{Barry14} and ours are unrelated.

\textbf{\textit{Organisation of the paper}}: The rest of this paper is organised as follows. Section \ref{DEx} gives formal definitions of qMDPs and their reachability probabilities and invariant subspaces. It also presents several examples to illustrate how can quantum algorithms and protocols be modelled as qMDPs and to show some essential differences between qMDPs and classical MDPs as well as quantum Markov chains. All main results obtained in the paper are stated in Section \ref{StatMain}. Sections \ref{PrFH} and \ref{PrIFH} are devoted to prove the results for finite-horizon and infinite-horizon, respectively. A brief conclusion is drawn in Section \ref{Concl}.

\section{Definitions and Examples}\label{DEx}
\subsection{Basics of Quantum Theory} For convenience of the reader, we very briefly recall some basic notions in quantum theory with the main aim being fixing notations; see \cite{NC00} for details.
In this paper we always assume that the state Hilbert space is $d-$dimensional, i.e. $\hs=\mathbb{C}^d$ where $\mathbb{C}$ is the field of complex numbers. We use the Dirac notation and assume that $\{|i\>\}_{i=1}^{d}$ is an orthonormal basis of $\hs$. Then we have $\hs=\spans\{|i\>\}$, a pure state in $\hs$ can be written as $|\psi\>=\sum \alpha_i|i\>$ with $\sum |\alpha_i|^2=1$, and a mixed state is represented by a density matrix in $\hs$, i.e. a semi-definite positive $d\times d$ matrix with trace $1$. Write $\mathcal{D}(\hs)$ for the set of all density matrices in $\hs$. The identity matrix is denoted $I$. If a density matrix can be written as $\rho=\sum p_i|\psi_i\>\<\psi_i|$, where $\<\psi_i|$ stands for the transpose conjugate of $|\psi_i\>$, then its support is $\supp(\rho) = \spans\{|\psi_i\>:p_i>0\}$.

The evolution of a closed quantum system is described by a $d\times d$ unitary matrix: $|\phi\>\mapsto U|\phi\>$. A super-operator $\E:\mathcal{D}(\hs)\ra\mathcal{D}(\hs)$ depicts the dynamics of a system which is realised with noise or interacts with its environment, and it can always be represented by $\E(\rho)= \sum E_i \rho E_i^\dag$ where all $E_i$ are $d\times d$ matrices with $\sum E_i^\dag E_i=I$ and $E_i^\dag$ denotes the conjugate transpose of $E_i$. The $d^2\times d^2$ matrix $M=\sum (E_i\otimes E_i^\ast)$ is called the matrix representation of $\E$.

A quantum measurement in $\hs$ is described by a set of $d\times d$ matrices $M=\{M_{m_1},\cdots,M_{m_k}\}$ with $\sum M_{m_i}^\dag M_{m_i}=I$, where $m_i$'s denote the possible outcomes. If we perform measurement $M$ on a quantum system which is currently in state $\rho$, then the probability that we get outcome $m_i$ is $p_i=\tr(M_{m_i}^\dag M_{m_i}\rho)$ and the system's after-measurement state is $\rho_i=M_{m_i}\rho M_{m_i}^\dag/p_i$ whenever the outcome is $m_i$.
A measurement $P=\{P_{m_1},\cdots,P_{m_k}\}$ is projective if $P_{m_i}P_{m_j} = \delta_{ij}P_{m_i}$.

\subsection{Quantum Markov Decision Processes}
In this subsection, we formally define our notions of qMDPs and their schedulers.
\begin{defi}\label{defi:qmdp} A qMDP is a 4-tuple $\mathcal{M}=\langle \hs,Act,\mathbf{M}, Q\>$, where:
    \begin{itemize}
      \item $\hs$ is a d-dimensional Hilbert space, called the state space. The dimension of $\hs$ is also called the dimension of $\mathcal{M}$, i.e. $\dim\mathcal{M}=\dim\hs=d$.
      \item $Act$ is a finite set of action names. For each $\alpha\in Act$, there is a corresponding super-operator $\E_\alpha$ that is used to describe the evolution of the system caused by action $\alpha$.
      \item $\mathbf{M}$ is a finite set of quantum measurements. We write $\Omega$ for the set of all possible observations; that is,
        $$\Omega = \{O_{M,m}: M\in \mathbf{M} \text{~and~} m \text{~is\ a\ possible\ outcome\ of~} M\}.$$ Intuitively, $O_{M,m}$ indicates that we perform the measurement $M$ on the system and obtain the outcome $m$.
      \item $Q: Act\cup \mem \ra 2^{Act\cup \mem}$ is a mapping. For each $\alpha\in Act$ (or $M\in\mathbf{M}$), $Q(\alpha)$ (resp. $Q(M)$) stands for the set of the available actions or measurements after $\alpha$ (resp. $M$) is performed. For the trivial case that $Q(\alpha)=Act\cup \mem$ for all $\alpha$, $Q$ will be omitted, and the qMDP $\M$ will be simply written as a triple $\<\hs,Act,\mem\>$.
    \end{itemize}
\end{defi}

\begin{defi}
A scheduler for a qMDP $\mathcal{M}$ is a function $$\mathfrak{S}:(Act\cup\Omega)^{*}\ra Act\cup \mathbf{M}.$$ For any sequence $\sigma=\alpha_1...\alpha_n\in (Act\cup\Omega)^*$, $\mathfrak{S}(\sigma)$ indicates the next action or measurement after actions or observations $\alpha_1...\alpha_n$ happen.
\end{defi}

As pointed out in the introduction, a scheduler in a qMDP selects the next action based on the outcomes of performed measurements. Actually, in the above definition the performed actions are also recorded as a part of the information for such a selection. This design decision is motivated by several examples in Subsection \ref{Ex-QMDP}. We now describe the evolution of a qMDP $\M$ with an initial state $\rho\in\mathcal{D}(\hs)$ and a scheduler $\s$. For simplicity, we write $W=(Act\cup \Omega)^*$. For each word $w\in W$,
the state $\rho_w^{\s}$ of the qMDP $\M$ and probability $p_w^\s$ that this state is reached in $\M$ after sequence $w$ of actions or observations are defined by induction on the length of $w$:
\begin{itemize}\item $\rho_\epsilon^\s=\rho$ and $p_\epsilon^\s=1$, where $\epsilon$ is the empty word. \item If $\s(w)=\alpha\in Act$, then $\rho_{w\alpha}^\s=\E_\alpha(\rho_w^\s)$ and $p_{w\alpha}^\s =p_w^s$. (Note that all the super-operators $\E_\alpha$ $(\alpha\in Act)$ are assumed to be trace-preserving.)
\item If $\s(w)=\alpha=O_{M,m}\in\Omega$, then $$\rho_{w\alpha}^\s = M_m \rho_w^\s M_m^\dag/\tr(M_m \rho_w^\s M_m^\dag)$$ and $p_{w\alpha}^\s = p_w^\s\cdot\tr(M_m \rho_w^\s M_m^\dag).$
\end{itemize}
Furthermore, for each $n\geq 0$, we can define the global state of the qMDP $\M$ at step $n$ according to scheduler $\s$ by $$\rho(n,\s) = \sum_{w\in W\ {\rm s.t.}\ |w|=n} p_w^\s\rho_w^\s.$$ For a subspace $B$ of $\hs$, the probability that $B$ is reached at step $n$ with initial state $\rho$ and scheduler $\s$ is defined by \begin{equation}\label{repr}\Pr(\rho(n,\s)\vDash B)=\tr(P_B\rho(n,\s))\end{equation}where $P_B$ is the projection onto $B$.

\subsection{Invariant Subspaces}

A key notion used in reachability analysis of quantum Markov chains \cite{YFYY}  is BSCC. A counterpart of BSCC in qMDPs is the notion of (common) invariant subspace. Let $B$ be a subspace of Hilbert space $\hs$. We say that $B$ is invariant under a super-operator $\E$ if $\supp(\E(\rho))\subseteq B$ for all density matrices $\rho$ with $\supp(\rho)\subseteq B$. Moreover, $B$ is invariant under a measurement $M=\{M_1,\cdots,M_k\}$ if $\supp(M_i\rho M_i^\dag)\subseteq B$ for all $1\leq i\leq k$ and all $\rho$ with $\supp(\rho)\subseteq B$.

\begin{defi} Let $\M=\langle\hs,Act,\mathbf{M},Q\rangle$ be a qMDP and $B$ a subspace of $\hs$. If $B$ is invariant under super-operator $\E_\alpha$ for all $\alpha\in Act$, and it is invariant under all measurement $M\in\mathbf{M}$, then $B$ is called an invariant subspace of $\M$.\end{defi}

The probability that an invariant subspace is reached is a non-decreasing function of the number of steps.
\begin{thm}\label{thm:PrInvInc}
    Let $\M$ be a qMDP with initial state $\rho$ and $B$ an invariant subspace of $\M$. Then for any scheduler $\s$ and $n\geq 0$, we have:
    $$ \Pr(\rho(n+1,\s)\vDash B)\geq \Pr(\rho(n,\s)\vDash B). $$
\end{thm}
\begin{proof}
Induction on $n$ by using Theorem 1 in \cite{YFYY}.
\end{proof}

\subsection{Reachability Probability}
The reachability probability of finite-horizon was defined in equation (\ref{repr}). Now we define the reachability probability of infinite-horizon.

\begin{defi}\label{defi:supremum}
Let $\M$ be a qMDP with state Hilbert space $\hs$, $\rho$ an initial state, $\s$ a scheduler for $\M$ and $B$ a subspace of $\hs$. Then reachability probability of $B$ in $\M$ starting in $\rho$ with scheduler $\s$ is defined by \begin{equation}\label{Recp}\Pr\nolimits^{\s}(\rho\vDash\Diamond B)=\lim_{n\ra\infty}\Pr(\rho(n,\s)\vDash B).\end{equation}
\end{defi}

It is worth noting that, in general, the limit in the above equation does not necessarily exist. However, we have:

\begin{lem}\label{lem:InvPrLim} If $B$ is an invariant subspace of $\M$, then for any initial state $\rho$ and any scheduler $\s$, the reachability probability
    $\Pr\nolimits^\s(\rho\vDash\Diamond B)$ always exists.
\end{lem}

\begin{proof}Since $\Pr(\rho(n,\s)\vDash B)$ is bounded by $1$, the conclusion follows immediately from Theorem \ref{thm:PrInvInc}. \end{proof}

\begin{defi}Let $\M$ be a qMDP with state Hilbert space $\hs$, $\rho$ an initial state and $B$ a subspace of $\hs$. Then supremum reachability probability of $B$ in $\M$ starting in $\rho$ is defined by \begin{equation}\label{Max-def}\Pr\nolimits^{\sup}(\rho\vDash \Diamond B) = \sup_\s \Pr\nolimits^\s(\rho\vDash\Diamond B).\end{equation} If scheduler $\s_0$ satisfies that $\Pr\nolimits^{\s_0}(\rho\vDash \Diamond B)=\Pr\nolimits^{\max}(\rho\vDash \Diamond B)$, then $\s_0$ is called the optimal scheduler for the initial state $\rho$.\end{defi}

\subsection{A Difference between Classical and Quantum Markov Decision Processes}

It is well-known \cite[Lemma 10.102]{BK08} that there exists a memoryless scheduler $\s_0$ that is optimal for all initial states. In the quantum case, however, it is possible that no optimal scheduler exists even for a fixed initial state.

\begin{exam}\label{exam:diff1}
    Consider a quantum Markov decision process $\M=\<\hs,Act,\mathbf{M}\>$, where $\hs=\spans \{|1\rangle, |2\rangle, |3\rangle, |4\rangle\}$, $\mathbf{M}=\emptyset$,  $Act=\{\alpha,\beta\}$ and
    \begin{align*}
    \E_\alpha(\rho) &= (|2\>\<1|\rho|1\>\<2|+|1\>\<1|\rho|1\>\<1|)/2+|2\>\<2|\rho|2\>\<2|\\
    &+|3\>\<3|\rho|3\>\<3|+|4\>\<4|\rho|4\>\<4|,
    \end{align*}
    \begin{align*}
        \E_\beta(\rho) &= |4\>\<1|\rho|1\>\<4|+|3\>\<2|\rho|2\>\<3|+|3\>\<3|\rho|3\>\<3|\\&+|4\>\<4|\rho|4\>\<4|.
    \end{align*}
Let $\rho_0=|1\>\<1|$ and $B=\spans\{|3\rangle\}$. Then $$\Pr\nolimits^{\mathcal{P}}(\rho_0\vDash \Diamond B) < \sup_\s \Pr\nolimits^\s(\rho_0\vDash\Diamond B)=1$$ for all schedulers $\mathcal{P}$.
Indeed, if $\mathcal{P}=\alpha^\omega$, then $\Pr\nolimits^{\s}(\rho_0\vDash \Diamond B)=0$. Let $\mathcal{P}\neq \alpha^\omega$ be a scheduler and let $k$ be the first index such that $a_k=\beta$ where $\mathcal{P}=a_1 a_2 \dots$. Then
$\Pr\nolimits^{\mathcal{P}}(\rho_0\vDash \Diamond B) = 1-0.5^{k-1}<1.$\end{exam}

One reason for nonexistence of the optimal scheduler is that the current state of a quantum system usually cannot be known exactly from the outside, and thus we often have no enough information to choose the next action in a scheduler for a qMDP. In the above example, whence we know the exact state of the system, we can choose an appropriate action to reach the target state: if the state is $|1\rangle$, we take $\alpha$, and if the state is $|2\rangle$, we take $\beta$. However, consider the case where the first action is $\alpha$. The state of the system will become $\rho_1=(|1\>\<1|+|2\>\<2|)/2$. Then we do not know it is in $|1\>$ or $|2\>$ exactly, and we cannot decide which action should be taken.

However, the above is not the only reason for nonexistence of the optimal scheduler. As shown in the following example, it is still possible that a qMDP has no the optimal scheduler when we know exactly its state.
\begin{exam}
    Let $\M=\<\hs,Act,\mem\>$ be a qMDP, $\rho_0=|1\>\<1|$ an initial state and $B=\spans\{|4\>\}$, where
    \begin{itemize}
      \item $\hs=\spans\{|1\>,|2\rangle,|3\rangle,|4\>\}$  ;
      \item $Act=\{a,b\}$ and $\mem=\emptyset$;
      \item $\E_a=A_1\cdot A_1^\dag + A_2\cdot A_2^\dag+ A_3\cdot A_3^\dag$, where
        \begin{equation*}
            A_1 = \begin{pmatrix}
                \cos\theta & \sin\theta & &\\
                -\sin\theta & \cos\theta & &\\
                & & 0 & \\
                & & & 0
            \end{pmatrix},
        \end{equation*}
         $\theta =0.6$, $A_2=|3\>\<3|$ and $A_3=|4\>\<4|$;
      \item $\E_b = \sum_{i=1}^4 C_i\cdot C_i^\dag$, where $C_1=|3\>\<1|,C_2=|4\>\<2|,C_3=|3\>\<3|,C_4=|4\>\<4|$.
    \end{itemize}
    Since $\theta=0.6$, the set $\{A_1^n|1\>:n\in \mathds{N}\}$ is dense on the circle $\{a|1\>+b|2\>:a,b\in \mathds{R}, a^2+b^2=1\}$. For any $\epsilon>0$, there exists $n$, such that $\E_a^n(|1\>\<1|)=|\psi_n\>\<\psi_n|$ with $|\<2|\psi_n\>|>\sqrt{1-\epsilon}$. Thus $\Pr^\s(\rho_0\vDash B)>1-\epsilon$ for $\s=(a^nb)^\omega$. This leads to $\Pr^{\sup}(\rho_0\vDash B) =1$. But since $A_1^m|1\>\neq |2\>$ for any $m$, there is no optimal scheduler.
\end{exam}
In the above example, we have complete information about the state of the system after $\E_a$: it is always a superposition $a|1\>+b|2\>$ of $|1\>$, $|2\>$. But this does not help to derive an optimal scheduler because only $|2\>$ can reach the supremum 1.

\subsection{A Difference between Quantum Markov Chains and Decision Processes}
It was shown in \cite{YFYY} that a quantum Markov chain will eventually reach a subspace $B$ for any initial state if there is no BSCC contained in the ortho-complement $B^\perp$. The following question asks whether a similar conclusion holds for qMDPs.
\begin{prob}\label{ques1}
    Let $\M$ be a qMDP with state space $\hs$, $\s$ a given scheduler for $\M$ and $B$ a subspace of $\hs$. Suppose that $\M$ has no invariant subspace contained in $B^\perp$. Will $\M$ reach $B$ eventually, i.e. $\Pr^\s(\rho\vDash\Diamond B)=1$ for all initial states $\rho$?
\end{prob}
This question is negatively answered by the following example.

\begin{exam}
    Let $\M=\<\hs,Act,\mathbf{M}\>$ with $\hs=\spans\{|1\>,|2\>,|3\>\}$, $Act=\{a,b\}$ and $\mathbf{M}=\emptyset$. The super-operators corresponding to $a$ and $b$ are defined as follows:
    \begin{align*}
        \E_a(\rho) &= |3\>\<1|\rho|1\>\<3|+|1\>\<2|\rho|2\>\<1|
        +|3\>\<3|\rho|3\>\<3|,
    \end{align*}
    \begin{align*}
        \E_b(\rho) &= |2\>\<1|\rho|1\>\<2|+|3\>\<2|\rho|2\>\<3|
        +|3\>\<3|\rho|3\>\<3|
    \end{align*}
    for any density operator $\rho$. Let $B=\spans\{|3\>\}$. It is easy to see that $\E_a$ and $\E_b$ have no common invariant subspace in $B^\perp$. We consider initial state $\rho_0=(|1\>\<1|+|2\>\<2|)/2$ and two schedulers $\s_1=(ab)^\omega$ and $\s_2=(ab)^kaa(ab)^\omega$ for some $k$. Then we have $\Pr^{\s_1}(\rho_0\vDash \Diamond B)=1/2$, but $\Pr^{\s_2}(\rho_0\vDash \Diamond B)=1$.
\end{exam}

\subsection{Quantum Algorithms and Protocols as qMDPs}\label{Ex-QMDP}
In this subsection, we show how can the existing quantum algorithms and communication protocols be seen as examples of qMDP by analysing their structures.
The early quantum algorithms and protocols can be roughly classified into three classes:\begin{enumerate}\item
The first class  applies a sequence of unitary operators followed by a measurement. If the outcome of measurement is desirable, the algorithm terminates. Otherwise, the algorithm is re-initialized and executed again; see Figure \ref{Fig:MExam0a}. Examples include the famous quantum order-finding and factoring algorithms \cite{NC00}, the Grover search algorithm \cite{Grover96}, several quantum-walk-based algorithms \cite{SKW03,Childs03,Kempe05} and the algorithm for solving the expectation value of some operators of systems of linear equations \cite{HHL09}.
\item The second class repeatedly applies an action-measurement loop until success; see Figure \ref{Fig:MExam0b}. One example is the routing algorithm based on a many-measurement quantum walk in \cite{Kempe05}.
\item The structure of the third class looks like a decision tree; see Figure \ref{Fig:MExam2}. Examples are quantum teleportation \cite{NC00}, one-way quantum computer \cite{RB01}. These examples always terminate.
\end{enumerate}

\begin{figure}[!ht]
    \centering
    \subfigure[]{
    \includegraphics[height=2cm]{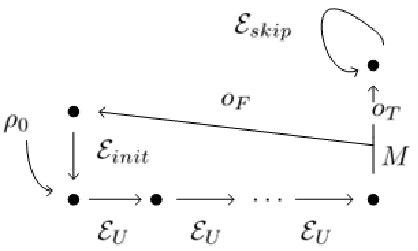}\label{Fig:MExam0a}
    }
    \subfigure[]{
    \includegraphics[height=2cm]{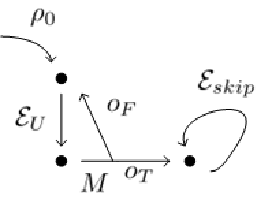}\label{Fig:MExam0b}
    }
    \caption{$\E_U$ represents the one or several sequential unitary operators. $\rho_0$ is the initial state. $\E_{init}$ represents re-initializing, i.e., restarting the algorithm. $\E_{skip}$ means maintaining the result for further application. $M$ represents measurements with observation $o_T$ standing for success and $o_F$ for failure.}\label{Fig:MExam0}
\end{figure}

\begin{figure}[!ht]
    \begin{center}
    \includegraphics[height=3cm]{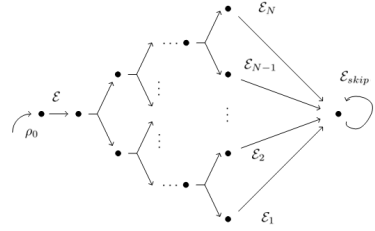}
    \caption{Structure of quantum decision trees.}\label{Fig:MExam2}
    \end{center}
\end{figure}

Recently, several algorithms have been developed with the structures different from Figures \ref{Fig:MExam0} and \ref{Fig:MExam2}. For example, a modified quantum factoring algorithm was experimentally realised in \cite{MLLAZB12}, where in order to reduce the number of necessary entangled qubits, the ancilla (control) qubits are recycled. The structure of this algorithm is shown in Figure \ref{Fig:MExamFactoring}. Another example is the quantum Metropolis sampling \cite{TOVPV11}. This algorithm can be used to prepare the ground or thermal state of a quantum system. The structure of this algorithm for reaching the ground state is shown in Figure \ref{Fig:MExamMetropolis}. It consists of decisions dependent on the history of actions and measurement outcomes as well as repeated loops until success.

As indicated by Figures  \ref{Fig:MExam0}-\ref{Fig:MExamMetropolis}, all of the algorithms and protocols mentioned above can be seen as qMDPs. Here we only elaborate the qMDP model of quantum Metropolis sampling.

\begin{figure}[!ht]
    \begin{center}
    \includegraphics[height=2cm]{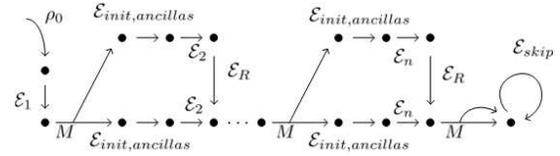}
    \caption{Structure of modified quantum factoring algorithm.}\label{Fig:MExamFactoring}
    \end{center}
\end{figure}

\begin{figure}[!ht]
    \centering
    \subfigure[The global view]{
    \includegraphics[height=2.5cm]{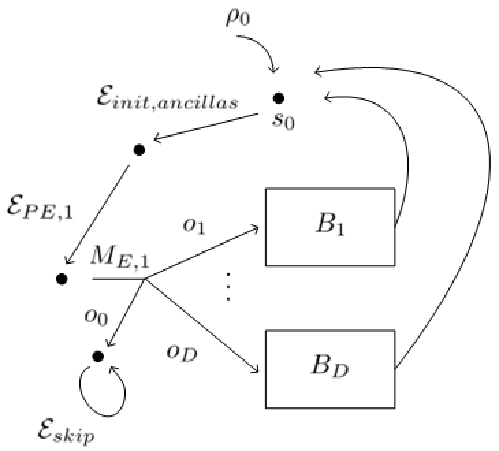}\label{Fig:MExamMG}
    }
    \subfigure[The structure of $B_i$]{
    \includegraphics[height=2.5cm]{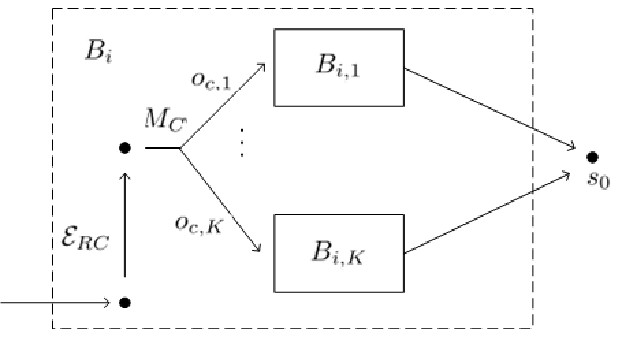}\label{Fig:MExamMBi}
    }
    \subfigure[The structure of $B_{ij}$]{
    \includegraphics[height=2.5cm]{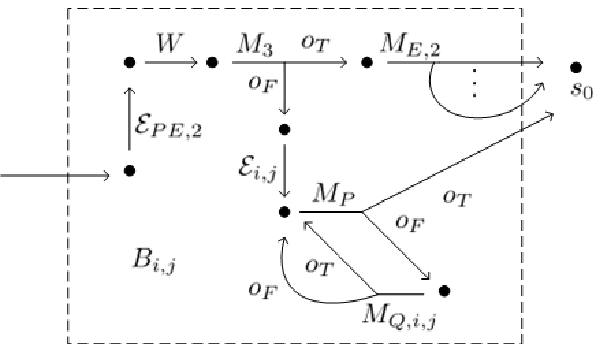}\label{Fig:MExamMBij}
    }
    \caption{Structure of quantum Metropolis sampling in \cite{TOVPV11}.}\label{Fig:MExamMetropolis}
\end{figure}

\begin{exam}\label{exam:MExamMetropolis} The qMDP $\M = \<\hs,Act,\mathbf{M}\>$ for the quantum Metropolis algorithm \cite{TOVPV11} is defined as follows:
    \begin{itemize}
      \item The state Hilbert space is the tensor product of five spaces, $\hs=\hs_{S}\otimes\hs_{E1}\otimes\hs_{E2}\otimes\hs_{a}\otimes\hs_{c}$, where
          \begin{enumerate}
            \item $\hs_S$ is the Hilbert space of the original system, whose ground state is the target.
            \item $\hs_{E1}$ and $\hs_{E2}$ are ancilla spaces, used to represent the energies of the states in $\hs_S$, where $\hs_{E1}$ represents the energy before updating in each round and $\hs_{E2}$ represents the new energy after updating.
            \item $\hs_a$ is $2-$dimensional with its basis states represent the success or failure of eigenstate updating.
            \item $\hs_c$ is used to implement the probabilistic choice of actions $C$.
          \end{enumerate}
      \item $Act$ consists of actions in the form of $\E_*$ in Figure \ref{Fig:MExamMetropolis}, where $\E_{RC}$ stands for probabilistic choice of unitary operators $C$ in \cite{TOVPV11}.
      \item $\mem$ consists of measurements in the form of $M_*$ in Figure \ref{Fig:MExamMetropolis}. $\Omega$ is the set of observations.
    \end{itemize} The task of the algorithm is actually to find a scheduler that reaches the ground state in this qMDP. One such scheduler is illustrated in Figure \ref{Fig:MExamMetropolis}.
\end{exam}

Various generalisations and variants of quantum Metropolis sampling have been proposed, e.g. quantum rejection sampling \cite{Ozols12}, quantum-quantum Metropolis sampling \cite{Yung12} and complementing quantum Metropolis algorithm \cite{Riera12}. An experiment for preparing thermal states was realised \cite{Zhang12} by employing some ideas from quantum Metropolis sampling.
The correctness of quantum Metropolis algorithm and its variants can actually be seen as a reachability problem for qMDPs. This motivates us to systematically develop techniques for reachability analysis of qMDPs.

\subsection{A Concurrent Quantum Program}
As one more example of qMDP, we consider a simple concurrent quantum program consisting of $n$ processes. Every process is a quantum loop. We assume a yes/no measurement $M=\{P_0,P_1\}$ in the state Hilbert space $\hs$, which is projective; that is, both $P_0$ and $P_1$ are projections. For each $1\leq i\leq n$, the $i$th process behaves as follows: it performs the measurement $M$, if the outcome is $0$, then it executes a unitary transformation $U_i$ and enter the loop again; if the outcome is $1$ then it terminates. Note that the loop guard (termination condition) of the $n$ processes are the same, but their loop bodies, namely unitary transformations $U_i$, are different.

This concurrent quantum program can be modelled as a qMDP $\M$ with $Act=\{1,2,...,n\}$. For each $i\in Act$, the action super-operator $\E_i$ is defined by $$\E_i(\rho)=P_1\rho P_1+U_i P_0\rho P_0 U_i^\dag$$ for all density matrices $\rho$. If $P_1$ is the projection onto the subspace $B$ of $\hs$, then the overall termination probability of the concurrent program with initial state $\rho$ is the supremum reachability $\Pr\nolimits^{\sup} (\rho\vDash \Diamond B).$ The following proposition provides us with a method for computing this termination probability.
We write $\bar{\E}$ for the average super-operator  of $\E_i$ ($1\leq i\leq n$); that is, $$\bar{\E} = \frac{1}{n} \sum_{i=1}^n \E_i.$$
We further define
 \begin{equation*}
        \E_\infty = \lim_{N\ra\infty} \frac{1}{N}\sum_{i=1}^{N} \bar{\E}^i.
 \end{equation*} (It was shown in \cite{Y13} that $\E_\infty$ can be computed by Jordan decomposition of the matrix representation of $\bar{\E}$.)

\begin{prop} \begin{enumerate}\item The overall termination probability $$\Pr\nolimits^{\sup} (\rho\vDash \Diamond B)=  1-\tr(\rho P_C),$$
where $C = \supp(\E_\infty(I_\hs))$ and $P_C$ is the projection onto $C$.
\item There is a string $s\in Act^*$ such that the scheduler $\s=s^\omega$ can attain the overall termination probability; that is, $$\Pr\nolimits^\s (\rho\vDash \Diamond B) = \Pr\nolimits^{\sup} (\rho\vDash \Diamond B).$$\end{enumerate}
\end{prop}
\begin{proof}
    Let $Y_a$ be an invariant subspace included in $B^\perp$ of $\E_a$. Since $Y_a\perp B$, we have $Y_a\supseteq \E_a(Y_a) = U_aY_a$. As $\dim Y_a = \dim (U_aY_a)$, we have $Y_a=U_aY_a$. Since unitary operators preserves the orthogonality, we have
    $\tr(\rho P_{Y_a})=\tr(\E_a(\rho) P_{Y_a}).$ If we write $C = \supp(\E_\infty(I_\hs))$, then $C$ is invariant by definition and we have $\Pr\nolimits^{\s'} (\rho\vDash \Diamond C)=  \tr(\rho P_C)$ for any scheduler $\s'$.
By Theorem \ref{thm:supOneSchAny} below, there exists $\s=s^\omega$ such that $\Pr\nolimits^\s (\rho\vDash \Diamond B\cup C)=1.$
    So,
    \begin{equation*}
        \Pr\nolimits^\s (\rho\vDash \Diamond B) = \Pr\nolimits^{\sup} (\rho\vDash \Diamond B)=  1-\tr(\rho P_C).
    \end{equation*}
\end{proof}

\section{Statement of Main Results}\label{StatMain}

The aim of this paper is to study decidability and complexity of reachability analysis for qMDPs. For readability, we summarise the main results in this section but postpone their proofs to the sequent sections.

\subsection{Results for the Finite-Horizon}\label{StatFH}

We first examine the case of finite-horizon and consider the following:
\begin{prob}\label{prob:cutpoint}
    Given a qMDP $\M$, an initial state $\rho$, a subspace $B$ of $\hs$ and $0\leq p\leq 1$, are there a scheduler $\s$ and a non-negative integer $n$ such that $$\Pr(\rho(n,\s)\vDash B) \triangle p$$ where $\triangle \in \{>,\geq,<,\leq\}$?
\end{prob}

\begin{thm}\label{thm:cutpoint}
    Problem \ref{prob:cutpoint} is undecidable for any $\triangle$.
\end{thm}

Now let us consider a qualitative variant of Problem \ref{prob:cutpoint}.
\begin{prob}\label{prob:spaceFinite}
    Given a qMDP $\M$ with the state Hilbert space $\hs$ and a subspace $B$ of $\hs$.
    \begin{enumerate}
      \item Are there a scheduler $\s$ and an integer $n$ such that $\supp(\rho(n,\s))$ $\subseteq B$ for all initial states $\rho$?
      \item Are there a scheduler $\s$ and an integer $n$ such that $\supp(\rho_0(n,\s))$ $\subseteq B$ for a given initial state $\rho_0$?
    \end{enumerate}
\end{prob}

The counterpart of Problem \ref{prob:spaceFinite}.2 for classical MDPs can be stated as follows: given a MDP $\M$ with a finite set $S$ of states, an initial state $s_0$ and $B\subseteq S$, decide whether there exists a scheduler $\s$ and an integer $N$ such that for any possible sequence of states $s_0 s_{1}s_2 \cdots $ under $\s$, there exists $j<N$ such that $s_{j}\in B$.
The polynomial-time decidability of this problem immediately follows the fact that an optimal scheduler for maximum reachability problem of a MDP can be found in polynomial time \cite{BK08}. The only thing we need to do is to check whether there exists a cycle in all states reachable from $s_0$ in $S\backslash B$, if $\Pr(s_0\vDash \Diamond B)=1$. The same result is true for the counterpart of Problem \ref{prob:spaceFinite}.1 for classical MDPs. This idea also applies to partially observable MDPs with a technique for reducing them to MDPs \cite{Baier08}.

However, undecidability of Problem \ref{prob:spaceFinite} was proved in \cite{BJKP05} if subspace $B$ is allowed to be not invariant. We prove undecidability of the problem for invariant subspace $B$ and thus significantly improve the main result of \cite{BJKP05}.

\begin{thm}\label{thm:undebabilitySpaceFiniteAnyInv}
    Both Problems \ref{prob:spaceFinite}.1 and \ref{prob:spaceFinite}.2 with $|Act|+|\mem|\geq 2$ and $B$ invariant are undecidable.
\end{thm}

\subsection{Results for the Infinite-Horizon}\label{StatIH}

Let us turn to the case of infinite-horizon. If the target subspace $B$ is allowed to be not an invariant subspace, then the limit in equation (\ref{Recp}) does not necessarily exists, and we consider the corresponding upper limit:
\begin{thm}\label{thm:UndecBnotInv}
    Given a qMDP $\M$, an initial state $\rho_0$ and a subspace $B$ (not necessarily invariant of $\M$), then it is undecidable to determine whether $$\sup_\s\limsup_{n\ra \infty} \Pr(\rho(n,\s)\vDash B)=1.$$
\end{thm}

In the remainder of this section, we only consider invariant subspace $B$ of $\M$, since the supremum reachability probability is not well-defined for those subspaces that are not invariant (see Definition \ref{defi:supremum} and Lemma \ref{lem:InvPrLim}). As for classical MDPs, a major reachability problem for qMDPs is the following:

\begin{prob}\label{prob:reachPrSup}
    Given a qMDP $\M$, an initial state $\rho_0$ and an invariant subspace $B$.
    \begin{enumerate}
      \item Decide whether $\Pr^{\sup}(\rho_0\vDash \Diamond B)=1$.
      \item Furthermore what is the exact value of $\Pr^{\sup}(\rho_0\vDash\Diamond B)$?
    \end{enumerate}
\end{prob}

\begin{thm}\label{thm:exphard}
    Given a qMDP $\M$, an initial state $\rho_0$ and an invariant subspace $B$.
    \begin{enumerate}
      \item It is EXPTIME-hard to decide whether $\Pr^{\sup}(\rho_0\vDash \Diamond B)=1$ even for $\M$ whose actions are all unitary.
      \item The value of $\Pr^{\sup}(\rho_0\vDash \Diamond B)$ is uncomputable, if $\Pr^{\sup}(\rho_0\vDash \Diamond B)<1$.
    \end{enumerate}
\end{thm}

For a special class of super-operators and measurements operators, Theorem \ref{thm:exphard}.1 can be significantly improved:
\begin{thm}\label{thm:supOneSp} Let $\M, \rho_0, B$ be as in Theorem \ref{thm:exphard}. We assume: \begin{enumerate}\item for each $\alpha\in Act$, $\E_\alpha(\rho)=\sum_{i=1}^{n_\alpha} A_{\alpha i}\rho A_{\alpha i}^\dag$ with all $A_{\alpha i}$ being of the form $a\cdot |\varphi\>\<\psi|$;
\item for $M_\beta\in\mem$, $M_\beta=\{M_{\beta 1},...,M_{\beta k_\beta}\}$ with all $M_{\beta j}$ being also of form $a\cdot |\varphi\>\<\psi|$.\end{enumerate}
Let $N=\max (\{n_\alpha:\alpha\in Act\}\cup\{k_\beta:M_\beta\in\mem\})$. Then whether $\Pr^{\sup}(\rho_0\vDash \Diamond B)=1$ can be decided in time $O(poly((|Act|+|\mem|)2^N))$.\end{thm}

A variant of Problem \ref{prob:reachPrSup} is the following:

\begin{prob}\label{prob:reachPrAnyRho}
    Given a qMDP $\M$, and an invariant subspace $B$, is there a scheduler $\s$, such that $\Pr^{\s}(\rho\vDash \Diamond B)=1$ for all initial states $\rho$?
\end{prob}

The difference between this problem and Problem \ref{prob:reachPrSup} is that the initial state is arbitrary in the former but it is fixed in the latter. It is worth noting that the counterparts of these two problems for classical MDPs are similar because they have only a finite number of states which can be checked one by one. However, the quantum versions are very different due to the fact that the state Hilbert space of a qMDP is a continuum. It is also worth carefully comparing this problem with Problem \ref{ques1}: scheduler $\s$ is given in the latter, whereas we want to find a special scheduler $\s$ in the former.

\begin{thm}\label{thm:supOneSchAny}
    For a given qMDP $\M=\langle\hs,Act,\mathbf{M}\rangle$ and an invariant subspace $B$ of $\M$, the following two statements are equivalent:
    \begin{enumerate}
        \item There exists a scheduler $\s$ such that $\Pr^\s(\rho\vDash\Diamond B)=1$ for all initial states $\rho$;
        \item There is no invariant subspace $C$ of $\M$ included in $B^\perp$.
    \end{enumerate}
    Furthermore, if there is no invariant subspace $C$ of $\M$ included in $B^\perp$, then there exists an optimal finite-memory scheduler $\s=s^\omega$ with $s\in Act^*$.
\end{thm}

Based on the above theorem, we develop Algorithm \ref{alg:findAScheduler} for checking existence of the optimal scheduler, of which the correctness and complexity are presented in the next theorem.

\begin{algorithm}
    \SetKwData{Left}{left}\SetKwData{This}{this}\SetKwData{Up}{up}
    \SetKwFunction{Union}{Union}\SetKwFunction{FindCompress}{FindCompress}
    \SetKwInOut{Input}{input}\SetKwInOut{Output}{output}

    \Input{A quantum Markov decision provess $\M$, the Hilbert space $\hs$, a subspace $B\subsetneq \hs$}
    \Output{A string $s\in (Act\cup \mem)^*$}
    * $s=\epsilon$ means no such scheduler.*\\
    \Begin{
        $s=\epsilon$\;
        $Act'\la Act$\;
        $E\la \E_i$, for all $i\in Act$\;
        $t \la |Act'|$\;
        \For{any $M\in\mem$}{
            $t\la t+1$\;
            $\E_t \la \sum M_i\cdot M_i^\dag$, for all measurement operators $M_i$ of $M$\;
            $E\la E\cup \{\E_t\}$\;
            $Act'\la Act'\cup\{t\}$\;
        }
        $\F=\frac{1}{t}\sum_{\E_i\in E} \E_i$\;
        \If{$\F(I_B)\not\subseteq B$}{
            \Return{$s$}\;
        }
        $\mathcal{G}\la \F|_{B^\perp}$\;
        $N \la \text{~null space of~} \mathcal{G}(x)-x=0$\;
        \If{$N!=\{0\}$}{
            \Return{$s$}\;
        }

        $d\la \dim(\hs)$\;
        $S \la \cup_{i=1}^{d-1}Act'^i$\;
        $T\la B^\perp$\;
        $b\la \dim(T)$\;

        \While{$b>0$}{
            \For{$v\in S$}{
                $w\la s\cdot v$\;
                $Y \la (\E_w^*(T^\perp))^\perp$\;
                \If{$\dim(Y)<b$}{
                    $b\la \dim(Y)$\;
                    $s\la w$\;
                    break\;
                }
            }
        }
        \Return{$s$}\;
    }
    \caption{Find an optimal scheduler}\label{alg:findAScheduler}
\end{algorithm}

\begin{thm}\label{thm:lag-correction}\begin{enumerate}\item Algorithm \ref{alg:findAScheduler} returns $\epsilon$ if there is no such a scheduler, otherwise it returns $s\neq \epsilon$ such that $\s=s^\omega$ is an optimal scheduler.
\item The time complexity of Algorithm \ref{alg:findAScheduler} is $O(d^9t^d)$, where $d=\dim \hs$ and $t=|Act|+|\mem|$.\end{enumerate}
\end{thm}

We now consider another variant of Problem \ref{prob:reachPrSup}, where not only the initial state $\rho$ but also the scheduler $\s$ can be arbitrary.

\begin{prob}\label{prob:reachPrAnyRhoAnySch}
    Given a qMDP $\M$, and an invariant subspace $B$, is the reachability probability always 1, i.e. $\Pr^{\s}(\rho\vDash \Diamond B)=1$ for all initial states $\rho$ and all schedulers $\s$?
\end{prob}

For this problem, we only have an answer in a special case.

\begin{thm}\label{thm:AnySchSuperoperator}
    Let $\M$ be a qMDP with $\mem=\emptyset$ and $B$ an invariant subspace of $\hs$. Then
    $\Pr\nolimits^\s (\rho\vDash \Diamond B) = 1$ holds for all schedulers $\s$ and all initial states $\rho$ if and only if it holds for all initial states and all schedulers of the form $\s=s^\omega$ with $|s|\leq L_d$, where $L_d$ is inductively defined as follows:
    \begin{itemize}
      \item  $L_0=1$ and $K_0=k$, where $k=|Act|$.
      \item $L_{i+1}=(K_i+1)L_i$ and $K_{i+1} = k^{L_{i+1}}$ for any $i\geq 0$.
    \end{itemize}
\end{thm}

We can develop an algorithm to check whether $\Pr\nolimits^\s (\rho\vDash \Diamond B)=1$ holds for all initial states and all schedulers $\s$. By Theorem \ref{thm:AnySchSuperoperator}, we only need to examine all schedulers of the form $\s=s^\omega$ with $|s|\leq L_d$. There are totally $K_d$ such schedulers, and for each one, it costs at most $O(d^6L_d)$ arithmetic operations to check the conclusion. Thus, the complexity of the algorithm is $O(d^6L_dK_d)$.
For the special class of qMDPs considered in Theorem \ref{thm:supOneSp}, we can significantly reduce this complexity.

\begin{thm}\label{thm:AnySchSuperoperatorSp}
    Let $\M$, $B$ and $N$ be as in Theorem \ref{thm:supOneSp}. Then whether $\Pr\nolimits^\s (\rho\vDash \Diamond B) = 1$ holds for all schedulers $\s$ and all initial states $\rho$ can be decided in time $O(poly((|Act|+|\mem|)2^N))$ .
\end{thm}

To conclude this section, we point out a link from Problems \ref{prob:reachPrAnyRho} and \ref{prob:reachPrAnyRhoAnySch} to a long-standing problem in matrix analysis and control theory, namely the joint spectral radius problem \cite{Gurvits95, TB97, CM05}. For a given set of square matrices $\Sigma=\{A_i:i=1,\cdots,t\}$, the discrete linear inclusion of $\Sigma$ is defined to be the set
\begin{equation*}\begin{split}
   DLI(\Sigma)= \{x_n=A_{s_n}\cdots A_{s_1}x_0&: x_0\in\mathds{C}^d, n\geq 0\\ &{\rm and}\ A_{s_j}\in \Sigma\ (1\leq j\leq n)\}.
\end{split}\end{equation*}
The set $\Sigma$ is said to be absolutely asymptotically stable (AAS) if $\lim_{n\ra \infty} A_{s_n}\cdots A_{s_1}=0$ for any infinite sequences $A_{s_1}A_{s_2}...$ in $\Sigma$.
The joint spectral radius and lower spectral radius of $\Sigma$ are defined as
\begin{equation*}
    \bar{\varrho}(\Sigma) = \limsup_{k\ra\infty} \bar{\varrho}_k(\Sigma),\ \ \ \ \ \ \underline{\varrho}(\Sigma) = \liminf_{k\ra\infty} \underline{\varrho}_k(\Sigma)
\end{equation*} respectively,
where for every $k\geq 1$, $$\bar{\varrho}_k(\Sigma)=\sup\{\|A\|^{1/k}:A\in\Sigma^k\},$$ $$\underline{\varrho}_k(\Sigma)=\inf\{\|A\|^{1/k}:A\in\Sigma^k\}.$$
It is known \cite{Gurvits95,CM05} that $\Sigma$ is AAS if and only if the joint spectral radius  $\bar{\varrho}(\Sigma)<1$.
It was shown in \cite{TB97} that unless $P=NP$, there are no polynomial-time approximate algorithms for computing $\bar{\varrho}(\Sigma)$. The problem ``$\underline{\varrho}(\Sigma)<1$'' and ``$\bar{\varrho}(\Sigma)\leq 1$'' were proved to be undecidable in \cite{TB97, BT00}. However, the problem whether ``$\bar{\varrho}(\Sigma)<1$'' is decidable is still open although the notion of joint spectral radius was introduced more than fifty years ago \cite{RS60}.

\begin{thm}\label{lem:toRadius} Let $\M$ be a qMDP with $\mem=\emptyset$ and $B$ an invariant subspace of $\M$. For each $\alpha\in Act$, let $M_\alpha$ be the matrix representation of $P_T\E_\alpha(\cdot) P_T$, where $T=B^\perp$. We write $\Sigma_\M=\{M_\alpha: \alpha\in Act\}$. Then:
    \begin{enumerate}
      \item $\underline{\varrho}(\Sigma_\M)<1$ if and only if there exists a scheduler $\s$ such that for any initial state $\rho$, it holds that $\Pr^\s(\rho\vDash\Diamond B)=1$.

      \item $\bar{\varrho}(\Sigma_\M)<1$ if and only if for any scheduler $\s$ and any initial state $\rho$, it holds that $\Pr^\s(\rho\vDash\Diamond B)=1$.
    \end{enumerate}
\end{thm}

\section{Finite-Horizon Problems}\label{PrFH}

In this section, we prove the theorems for finite-horizon stated in Subsection \ref{StatFH}.

\subsection{Proof of Theorem \ref{thm:cutpoint}}
We prove this theorem by an easy reduction from the emptiness problem of cut-point languages for probabilistic finite automata (PFA) to Problem \ref{prob:cutpoint}.
For a given MO-1gQFA $\M'=\<\hs,\Sigma,\rho_0,\{\E_\sigma\}_{\sigma\in\Sigma},P_{acc}\>$ \cite{Hirvensalo}, we can construct a qMDP $\M=\<\hs,Act,\mem\>$ such that $Act=\Sigma$, and $\mem=\Omega=\emptyset$. Let $B=\supp(P_{acc})$. Then these exist $\s$ and $n$ such that $\Pr(\rho(n,\s)\vDash B) \triangle p$ if and only if there exists a word $\sigma_1\cdots\sigma_n$ such that $\tr(P_{acc}\E_{\sigma_n}\circ\cdots\circ\E_{\sigma_1}(\rho_0))\triangle p$.
Since MO-1gQFA can simulate any PFA \cite{Hirvensalo} and the emptiness problem for PFA is undecidable \cite{BJKP05}, Problem \ref{prob:cutpoint} is undecidable too.

\subsection{Proof of Theorem \ref{thm:undebabilitySpaceFiniteAnyInv}}

 Our proof technique is a reduction from the matrix mortality problem to Problem \ref{prob:spaceFinite}. The matrix mortality problem can be simply stated as follows: \begin{itemize}\item Given a finite set of matrices  $G=\{M_i\in\mathds{Z}^{n\times n}: i\in\{1,2,\cdots,k\}\}$, is there any sequence $j_1,\cdots,j_m$ such that $M_{j_m}M_{j_{m-1}}\cdots M_{j_1}=0$?\end{itemize}
It is known \cite[Theorem 3.2]{Halava97} that the matrix mortality problem is undecidable for $k=2$.

We now prove Theorem \ref{thm:undebabilitySpaceFiniteAnyInv}. For a set $G$ of matrices as above, we construct a qMDP $\M=\langle\hs,Act,\mem\rangle$ from it as follows:
    \begin{itemize}\item The state space is $\hs=\spans\{|1\>,\cdots,|2n\>\}$.
      \item Let $Act=\{1,2,\cdots,n\}.$ For each $i\in Act$, we construct a super-operator $\E_i$ from $M_i$:$$\E_i(\rho) = A_i\rho A_i^\dag+B_i\rho B_i^\dag+C_i\rho C_i^\dag,$$ where
          $$A_i=\begin{pmatrix}
                a_i & 0\\
                0 & 0
          \end{pmatrix},
          B_i=\begin{pmatrix}
                0 & 0\\
                0 & I_{n\times n}
          \end{pmatrix},
          C_i=\begin{pmatrix}
                0 & 0\\
                c_i & 0
          \end{pmatrix},
          $$
          and $$a_i = \frac{1}{r_i}M_i, \ \ \ c_i=\sqrt{I-M_i^\dag M_i/r_i^2}.$$ In the defining equation of $a_i$, $r_i$ is a positive integer such that $I-M_i^\dag M_i/r_i^2\geq 0$.\item $\mem=\emptyset$.\end{itemize}
Now, it is easy to show that for any state
          $$\sigma=\begin{pmatrix}
                \sigma_a & *\\
                * & \sigma_b
          \end{pmatrix},$$ we have
          $$
          \E_i(\sigma) = \begin{pmatrix}
                a_i\sigma_a a_i^\dag & 0\\
                0 & \sigma_b+c_i\sigma_a c_i^\dag
          \end{pmatrix}.
          $$
          Therefore, for any initial state $$\rho_0=\begin{pmatrix}
                \rho_a & *\\
                * & \rho_b
          \end{pmatrix},$$ it holds that $$\rho(m,\s)=\begin{pmatrix}
                A\rho_a A^\dag & 0\\
                0 & *
          \end{pmatrix},$$ where $A = a_{j_m}\cdots a_{j_1}$. Now let $B = \spans\{|n+1\>,\cdots,|2n\>\}$.
          Then
          \begin{align*}
                \forall \rho_0, \exists n, \s ~s.t. &~ \rho(n,\s)\subseteq B\\
                \Leftrightarrow & \exists j_m,\cdots,j_1 ~s.t.~ a_{j_m}\cdots a_{j_1}=0\\
                \Leftrightarrow & \exists j_m,\cdots,j_1 ~s.t. ~ M_{j_m}\cdots M_{j_1}/(r_{j_m}\cdots r_{j_1})^2=0\\
                \Leftrightarrow & \exists j_m,\cdots,j_1 ~s.t. ~ M_{j_m}\cdots M_{j_1}=0.
          \end{align*}
          Since the matrix mortality problem is undecidable for $k = 2$, Problem \ref{prob:spaceFinite}.1 with $\mem=\emptyset$ and $B$ invariant is undecidable for dimension $|Act|+|\mem|\geq 2$.

Note in the above reduction, $A\rho_a A^\dag$ will always be rational, if $\rho_0$ is rational. Since for any $\sigma\geq 0$, $\supp(\sigma)\subseteq B$ holds if and only if $\sigma_a=0$, we only compute the upper left corner and leave $c_i$ as a symbol in the lower right corner when computing $\rho(m,\s)$. (There are at most $O(m)$ $c_i$'s.) Thus this reduction does not employ any operation on irrational numbers.

The reduction still works if we fix the initial state $\rho_0$ to be $I/2n$, which is one special case of Problem \ref{prob:spaceFinite}.2. Therefore Problem \ref{prob:spaceFinite}.2 is undecidable too.

\section{Infinite-Horizon Problems}\label{PrIFH}

In this section, we prove the theorems for infinite-horizon stated in Subsection \ref{StatIH}.

\subsection{Proof of Theorem \ref{thm:UndecBnotInv}}

This theorem can be proved by reduction from the value 1 problem of probabilistic automata on finite words in \cite{GO10}. The value 1 problem asks whether $\sup_{w\in \Sigma^*}\Pr(q_0\overset{w}{\ra} F)=1$ for a probabilistic finite automaton, where $q_0$ is the initial state, $F$ is the set of accept states and $w$ is a finite word over the input symbols $\Sigma$. We can reduce this automaton to a qMDP with $\mem=\emptyset$, $Act=\Sigma$, $\rho_0=|q_0\>\<q_0|$ and $B=\spans\{|q\>:q\in F\}$. The reduction technique is the similar as in the proof of Theorem \ref{thm:cutpoint}. Thus we have
\begin{equation}\label{eq:UndecBnotInv}
    \sup_{s\in Act^*}\Pr\nolimits^{s^\omega}(\rho(|s|,s^\omega)\vDash B)=1
\end{equation}
is undecidable. Since $\mem=\emptyset$, all schedulers are of form $\s=s^\omega$ or $\s\in Act^\omega$. Therefore equation \eqref{eq:UndecBnotInv} is equivalent to
\begin{equation*}
    \sup_{\s}\limsup_{n\ra\infty}\Pr\nolimits^{\s}(\rho(n,\s)\vDash B)=1.
\end{equation*}
This completes the proof.

\subsection{Proof of Theorem \ref{thm:exphard}}
We prove part 1 of the theorem by a reduction from an EXPTIME-complete game in \cite{SC79} to the problem of deciding whether $\Pr^{\sup}(\rho_0\vDash \Diamond B)=1$.
Some ideas are similar to those used in \cite{PT87,CDH10}.
\begin{itemize}
      \item  The game is a two-player game on a propositional formula $F(X,Y)$ in the conjunctive normal form (CNF). Player 1(resp. 2) changes \textit{at most one} variable in X (resp. Y) at each move, alternately. Once $F$ becomes \textit{true}, Player 1 wins.
    \end{itemize}

    It is known \cite{SC79} that the following problem is EXPTIME-complete: given an input string $w$ encoding a position of this game, decide whether Play 1 has a strategy to win definitely, where a position is a tuple $(\tau,F(X,Y),\alpha)$, where $\tau\in\{1,2\}$ denotes the current player, $F$ is a formula, and $\alpha$ is an assignment.

    Now we start to construct the reduction. Let $X=\{x_1,\cdots,x_n\}$, $Y=\{y_1,\cdots,y_n\}$, $\alpha\in \{0,1\}^{n+m}$ and $F=\bigwedge_{i=1}^{c} C_i$, where $C_i=\vee_{j=1}^{k} z_{i,j}$, and $z_{i,j}$ is one of $x_t,\neg x_t,y_t,\neg y_t$ for some $t$. We define a qMDP as follows:

    \textbf{State space}. The state space $\hs=\hs_S\otimes \hs_C \otimes \hs_F \otimes \hs_R$, where $\hs_S=\hs_2^{\otimes (n+m)}$, $\hs_C=\hs_{k+1}^{\otimes c}$, $\hs_F=\hs_{c+1}$, $\hs_R=\hs_{m+2}$, where $\hs_{i} =\spans\{|0\>,\cdots,|i\>\}$. The intuition behind the definition of these spaces is:
    \begin{itemize}
      \item $\hs_S$ encodes the assignment $\alpha$;
      \item $\hs_C$ is the work space for clauses;
      \item $\hs_F$ is the work space for the formula;
      \item $\hs_R$ encodes the randomness of Player 2's choice.
    \end{itemize}

    \textbf{Initial state}. The initial state is
    $|\psi_0\>=|\alpha(x_1)\>\cdots|\alpha(y_m)\>$ $|0_C\>|0_F\>|0_R\>$. We will see that the state of the system can always be represented in such a separable form during the computation of this qMDP.

    \textbf{$S_1$: Unitary operators for modelling actions by Player 1}. Since Player 1 can change at most 1 valuable, there are $n+1$ choices/actions:
    \begin{itemize}
      \item Do nothing: this can be described by the identity operator $I$;
      \item Change the $i$-th valuable $x_i$: this can be realised by the NOT gate $X=|0\>\<1|+|1\>\<0|$ operator on $i$-th space of $\hs_S$, i.e., $U_{1,i} = U_{x_i} \otimes I_{C} \otimes I_F \otimes I_R$, where
        \begin{equation*}
            U_{x_i} = I_{S,1}\otimes \cdots \otimes I_{S,i-1} \otimes X \otimes I_{S,i+1}\otimes\cdots\otimes I_{S,n+m}.
        \end{equation*}
    \end{itemize}
All these operators can be represented in this form using space $O(n(n+m+c+k))$.


    \textbf{$S_2$: Randomness of Player 2's choice}. First we split the state $|0\>\<0|$ in $\hs_R$ into $\frac{1}{\sqrt{m+1}}\sum_{i=1}^{m+1} |i\>$ by a unitary $$U_R=I_S\otimes I_C \otimes I_F\otimes ( \frac{1}{\sqrt{m+1}}\sum_{i=1}^{m+1} |i\>\<0|+\cdots ).$$ Then we apply
    \begin{align*}
        U_2= \sum_{i=1}^{m} U_{y,i}\otimes I_{C} &\otimes I_F \otimes |i\>\<i|+ I_S\otimes I_{C} \otimes I_F  \otimes\\ &|m+1\>\<m+1|
        + I_S\otimes I_{C} \otimes I_F  \otimes |0\>\<0|.
    \end{align*}
    At last, we apply a measurement $M_R=\{M_{R,i}:i=0,\cdots,m+1\}$, where
$M_{R,i} = I_S\otimes I_{C} \otimes I_F  \otimes |0\>\<i|.$
These step can be encoded in space $O(m^2 (n+m+c+k))$.

    \textbf{$S_C$: Checking the formula}. This can be done by the following steps:
    \begin{enumerate}
      \item First, we check each clause. A clause $C_i=\vee_{j=1}^{k} z_{i,j}$ is checked via each of its literals. For instance, if $z_{i,j}$ is $x_t$, we apply
          \begin{align*}
                U_{z,i,j} &= I_{S,1}\otimes \cdots \otimes I_{S,t-1} \otimes |1\>\<1| \otimes I_{S,t+1}\otimes\cdots\otimes I_{S,n+m}\\
                &\otimes U_{shift,i}\otimes I_F\otimes I_{R}\\
                &+I_{S,1}\otimes \cdots \otimes I_{S,t-1} \otimes |0\>\<0| \otimes I_{S,t+1}\otimes\cdots\otimes I_{S,n+m}\\
                &\otimes I_C\otimes I_F\otimes I_{R},
          \end{align*}
          where $U_{shift,i} =I_{C,1}\otimes\cdots\otimes I_{C,i-1}\otimes(\sum_{v=0}^{k-1} |v+1\>\<v|+|0\>\<k|)\otimes I_{C,t+1}\otimes\cdots\otimes I_{C,c}$ is the shift operator on subspace $\hs_{C_i}$. The case of $z_{i,j}$ being $\neg x_t, y_t,\neg y_t$ is similar. This step means that $z_{i,j}$ is true, and we shift one level in $\hs_{C_i}$.
      \item Second, we compute the value of the whole formula. This is similar the first step. If the state is $\hs_{C_i}$ is shifted at least once; that is, it is not $|0\>\<0|$, then we shift $\hs_F$ once.
      \item Third, we take a projective measurement $P_F=\{P_1,P_0\}$ on $\hs_F$, where
 $P_1 = I_S\otimes I_C \otimes |c\>\<c| \otimes I_R$ represents the fact that all $c$ clauses are true, i.e. $F$ is true, and
            \begin{equation*}
                P_0 = I_S\otimes I_C \otimes \sum_{i=0}^{c-1}|i\>\<i| \otimes I_R,
            \end{equation*}
            indicates that $F$ is false. If the outcome is $1$, we terminate.
      \item Forth, we undo the first two steps if the result is false. Let $U$ denote the unitary operator of the first two steps. If the projective measurement gives result $0$, the state remains unchanged because of the separable form of the initial state. Thus we can apply $U^\dag$ to undo $U$.
    \end{enumerate}
The above four steps can be represented in space $O(k^2 c(n+m+c+k))$.

    \textbf{Schedulers}. If the input $\tau=1$, i.e. Player 1 first moves, then we execute sequence $(S_1 S_C S_2 S_C)^\omega$ of steps; otherwise $(S_2 S_C S_1 S_C)^\omega$. This is realised by the mapping $Q$ in Definition \ref{defi:qmdp}. The decision is made in step $S_1$ (Player 1's turn).

    \textbf{Target and reachability probability}. The target is to reach the outcome $1$; that is, $P_1$ appears in $S_C$. Because of the separable form of the initial state, the state of the system is of the form $|\psi\>=|\alpha'(x_1)\>\cdots|\alpha'(y_m)\>|0_C\>|0_F\>|0_R\>$ after each step. Thus any step can be computed in polynomial time of $n,m,c,k$. Therefore, this is a polynomial time reduction. Furthermore, it is easy to see that Player 1 has a \textquotedblleft forced win\textquotedblright\ strategy if and only if there is a scheduler (for decisions in step $S_1$) with reachability probability is 1.

    \textit{Remark}: The target space $B$ may not be invariant. But we can easily modify the space $\hs_F$ so that $B$ becomes invariant. What we need to do is:
    \begin{itemize}
      \item extend $\hs_F$ to $k+2$ level space;
      \item change $P_1$ to $|k+1\>\<k|$, and add $P_2=|k+1\>\<k+1|$;
      \item make all unitary operators to be a controlled operator by $\hs_F$.
    \end{itemize}
    After the modification, the system state remains unchanged in each decision branching unless it reaches the target.

    We now turn to prove part 2 of the theorem; that is, $\Pr^{\sup}(\rho_0\vDash \Diamond B)$ is uncomputable.
    This can be done simply by a reduction from probabilistic automata on infinite words. In \cite{CH10}, it was shown that the following quantitative value problem is undecidable: for any $\epsilon>0$, does there exist a word $w$ such that the reachability probability in acceptance absorbing automata is greater than $\lambda-\epsilon$, for a given rational number $0<\lambda<1$.
 We reduce this problem to the supremum reachability problem for qMDPs. The reduction technique is similar to the proofs of Theorems \ref{thm:cutpoint} and \ref{thm:UndecBnotInv}. Since the automata are acceptance absorbing, $B$ is invariant. Thus, it is undecidable whether there exists $\s\in Act^\omega$, such that $\Pr^{\s}(\rho_0\vDash\Diamond B)>\lambda-\epsilon$. Since this is equivalent to decide $\sup_\s\Pr^\s(\rho_0\vDash \Diamond B)=\lambda$, we complete the proof.

\subsection{Proof of Theorem \ref{thm:supOneSp}}\label{sec:proofThmSupOneSp}

By the assumption, $\E_\alpha$ can be written as $\E_\alpha(\rho)=\sum_{i=1}^{n_\alpha} a_{\alpha,i}A_{\alpha,i}\rho$ $A_{\alpha,i}^\dag$, where $A_{\alpha,i}=a_{\alpha,i}|\varphi_{\alpha,i}\>\<\psi_{\alpha,i}|$. Then for any state $\rho$, we have $\E_\alpha(\rho)=\sum_i c_{\alpha,i}|\varphi_{\alpha,i}\>\<\varphi_{\alpha,i}|$ for some $c_{\alpha,i}
\geq 0$.
Define $Y(\alpha,\rho)\triangleq\supp(\E_\alpha(\rho))=\spans\{|\varphi_{\alpha,i}\>:c_{\alpha,i}>0\}$. It is easy to see that there are at most $2^{n_\alpha}$ different $Y(\alpha,\rho)$'s ranging over all $\rho$ for an given $\alpha$. Then the total number of $Y(\alpha,\rho)$'s with all actions $\alpha$ is at most $|Act|2^N$.
Similarly, we define $Z(\beta,j,\rho)\triangleq\supp(M_{\beta j}\rho M_{\beta j}^\dag)$. If probability $\tr(M_{\beta j}\rho M_{\beta j}^\dag)>0$, then $Z(\beta,j,\rho)$ $=\spans\{|\varphi_{\beta j}\>\}$. Otherwise it equals $\{0\}$. The total number of $Z(\beta,i,\rho)$'s is at most $|\mem|N$.
Thus there are at most $(|Act|+|\mem|)2^N$ possible different supports of resulting states. Let $Y$ to be the set of all these supports. Now we reduce this problem to the supremum-1 reachability problem of a classical Markov decision process $\M'=\<S,Act',T,s_0\>$:
\begin{itemize}
  \item each state corresponds to a possible support, i.e. $S=\{s_y:y\in Y\}\cup \{s_0\}$;
  \item $s_0 = s_{\{\supp(\rho_0)\}}$;
  \item $Act'= Act\cup\{\beta:M_\beta\in\mem\}$;
  \item for each $\alpha\in Act$, the transition function $T$ maps $s_x$ to $s_y$ with probability 1, where $\E_\alpha(x)=y$;
  \item for each $M_\beta\in\mem$, $T$ maps $s_x$ to $s_y$ with probability $1/l(\beta,x)$, where $y\in \{z: z=Z(\beta,i,x)\}$ and $l(\beta,y)$ is the number of elements in this set;
  \item the target states $B'=\{s_y:y\subseteq B \wedge y\in Y\}$.
\end{itemize}
For this classical Markov decision process, it is known \cite{BK08} that there is an optimal memoryless scheduler $\s_0$ such that $$P_{\max}\triangleq\Pr\nolimits^{\s_0}(s_0\vDash\Diamond B')=\Pr\nolimits^{\sup}(s_0\vDash\Diamond B').$$
If $P_{\max}=1$, then $\s_0$ can be converted to a scheduler of $\M$, whose decisions are based on supports of states. We immediately have $\Pr\nolimits^{\sup}(\rho_0\vDash\Diamond B)=1$.
Conversely, if $\Pr\nolimits^{\sup}(\rho_0\vDash\Diamond B)=1$, then for any $\epsilon>0$ there exists a history-dependent scheduler $\s_\epsilon$ convertible to that of $\M'$ such that $\Pr\nolimits^{\s_\epsilon}(s_0\vDash\Diamond B')>1-\epsilon$. Thus $P_{\max}=1$. This completes the reduction. The proof is finished by the fact from \cite{BK08} that the maximum reachability of a classical MDP can be solved in polynomial time of the size of $\M'$.

\subsection{Proofs of Theorems \ref{thm:supOneSchAny} and \ref{thm:lag-correction}}\label{sec:proofsProblemAnyRho}

We first present several technical lemmas. For a super-operator $\E$, we define:
\begin{equation}\label{subspace-def}
    X_\E = \spans \left(\bigcup \{\supp(\rho):\tr(P_B\E(\rho))=0\}\right).
\end{equation}
Since $B$ is invariant, $X_\E$ is obviously a subspace of $B^\perp$.

\begin{lem}\label{lem:Xspace}
    For any density operator $\rho$, $\tr(P_B\E(\rho))=0$ if and only if $\supp(\rho)\subseteq X_\E$.
\end{lem}
\begin{proof}
    The ``only if'' part is by definition. We now prove the ``if'' part. If $\supp(\rho)\subseteq X_\E$, then there exist $\sigma_1,\cdots,\sigma_k$ with $\supp(\sigma_i)\subseteq X_\E$ and $\supp(\rho)\subseteq\bigvee \supp(\sigma_i)$, i.e. $\rho\leq \gamma\sum\sigma_i$ for some $\gamma>0$. Thus
    \begin{align*}
        \gamma\sum\sigma_i-\rho\geq 0
         \Rightarrow& \E(\gamma\sum\sigma_i-\rho)\geq 0\\
        \Rightarrow & P_B \E(\gamma\sum\sigma_i-\rho) P_B \geq 0\\
        \Rightarrow & P_B \E(\gamma\sum\sigma_i) P_B \geq P_B\E(\rho)P_B.
    \end{align*}
    By definition, we have $P_B \E(\gamma\sum\sigma_i) P_B=0$ and $P_B\E(\rho)P_B\geq 0$. Therefore
    $P_B\E(\rho)P_B=0$. This implies $\supp(\rho)\subseteq X$.
\end{proof}

We now consider a special qMDP $\M=\langle\hs,Act,\mathbf{M}\rangle$ without measurements: $|\mem|=\emptyset$. We write $\E_s=\E_{s_k}\circ\cdots\circ\E_{s_2}\circ\E_{s_1}$ for a finite sequence $s=s_1s_2...s_k\in Act^*$. Since $\langle\hs, \E_s\rangle$ can be seen a quantum Markov chain, we know from  \cite{YFYY} that $\Pr^{\s} (\rho\vDash \Diamond B)=1$ if and only if there is no invariant subspace in $B^\perp$, where $\s=s^\omega$ is a periodic scheduler. For any $s\in Act^*$, we simply write $X_s$ for $X_{\E_s}$ defined by equation (\ref{subspace-def}) from super-operator $\E_s$.

\begin{lem}\label{lem:XspaceEmptySch}Let $\s=s^\omega$. If $X_s=\{0\}$, then $\Pr^\s(\rho\vDash\Diamond B)=1$ for any $\rho$.
\end{lem}
\begin{proof}
    We prove it by contradiction. Suppose $X_s=\{0\}$ and $\Pr^\s(\rho\vDash\Diamond B)<1$ for some $\rho$. Since $\langle\hs, \E_s\rangle$ is a quantum Markov chain, the scheduler $\s$ is a actually  repeated application of $\E_s$, we have from Theorems 4 and 6 in \cite{YFYY} that $\Pr^\s(\rho\vDash\Diamond B)<1$ if and only if there exists a (non-empty) BSCC $C$ of $T=B^\perp$ under $\E_s$. Corresponding to this BSCC, there exists a minimal fixed point state $\rho$ with $\E_s(\rho)=\rho$ and $\supp(\rho)=C\subseteq T=B^\perp$. By definition, we get $\{0\}\subsetneq \supp(\rho)\subseteq X_s=\{0\}$. A contradiction!
\end{proof}

\begin{lem}\label{lem:XsSub}
    For any $s,v\in Act^*$ and $w=sv$, we have $X_w \subseteq X_s.$
    In particular, if $\dim X_w= \dim X_s$, then we have $\E_s(X_s)\subseteq X_v$.
\end{lem}
\begin{proof}
    For any $\rho$ with $\supp(\rho)\subseteq X_w$, we have $0=P_B\E_w(\rho)P_B=P_B\E_v(\E_s(\rho))P_B$. Thus $\supp(\E_s(\rho))\subseteq X_v\subseteq B^\perp$. This implies $P_B\E_s(\rho)P_B=0$,  and $\supp(\rho)\subseteq X_s$. Therefore, it holds that $X_w \subseteq X_s$. We now turn to prove the second part. If $\dim X_w= \dim X_s$, then for any $\rho$ with $\supp(\rho)\subseteq X_s$, we have $\supp(\rho)\subseteq X_w$. This means $\supp(\E_s(\rho))\subseteq X_v$ as $B$ is invariant.
\end{proof}

Now we are ready to prove Theorems \ref{thm:supOneSchAny} and \ref{thm:lag-correction}.

{\renewcommand{\proofname}{Proof of Theorem \ref{thm:supOneSchAny}}
\begin{proof} The proof of (1) $\Rightarrow$ (2) is easy. Suppose that there is an invariant subspace $C$ of $\M$ included in $B^\perp$, then $\Pr^\s(\rho\vDash\Diamond B)=0$ for any $\rho$ in $C$ and for any scheduler $\s$.

We now prove (2) $\Rightarrow$ (1). For the special case of $\mathbf{M}=\emptyset$, assume that there is no invariant subspace $C$ of $\M$ included in $B^\perp$.
Let $D=\{\dim X_u: u\in Act^*\}$ and let $d_{\min} = \min D$. Then there exists $s\in Act^*$ such that $\dim X_s=d_{\min}$.
We assert that $d_{\min}=0$. Indeed, if $d_{\min}>0$, then for each word $v\in Act^*$, we put $w=sv$.
    By Lemma \ref{lem:XsSub}, we have $X_w\subseteq X_s$. Then it follows from the definition of  $d_{\min}$ that $X_w= X_s$. As a consequence, $Y\stackrel{\triangle}{=}\E_s(X_s)\subseteq X_v$. This implies $\E_v(Y)\subseteq B^\perp$. For a super-operator $\E$, we write $\R_{\E}(Y)$ for the transitive closure of $Y$ under $\E$, i.e. $$\R_{\E}(Y)\stackrel{\triangle}{=}\bigvee_{i=0}^{d-1}\E^i(Y).$$
    Let $\F=\frac{1}{t}\sum_{\alpha\in Act}\E_\alpha,$ where $t=|Act|$. We have:
    \begin{align*}
        \R_{\F}(Y) = \bigvee_{i=0}^{d-1}\bigvee_{x\in Act^i} {\E_x}(Y) \subseteq B^\perp.
    \end{align*}
It is clear that $\R_{\F}(Y)$ is invariant under $\F$, and thus invariant under any $\E_i$. So, $\R_{\F}(Y)$ is an invariant subspace of $\M$ included in $B^\perp$ under $\M$.This contradicts to the assumption. So, we have $d_{\min}=0$, and it follows from Lemma \ref{lem:XspaceEmptySch} that $\s=s^\omega$ is a optimal scheduler.

For the general case of $\mathbf{M}\neq\emptyset$, we define a super-operator $\E_{M_\beta}=\sum M_m\cdot M_m^\dag$ for each $M_\beta\in\mem$. Furthermore, we can construct a new qMDP $\M'=\<\hs,Act',\mathbf{M}'\>$ with $Act'=Act\cup \{\beta:M_\beta\in\mem\}$ and $\mathbf{M}'=\emptyset$. Then we complete the proof by applying the above argument to $\M'$.
\end{proof}
}
It is worth noting that the optimal scheduler given in the proof of the above theorem depends on which measurement is chosen in each step but not its outcome.
{\renewcommand{\proofname}{Proof of Theorem \ref{thm:lag-correction}}
\begin{proof}
The design idea of Algorithm \ref{alg:findAScheduler} is to see whether there exists an invariant subspace of $B^\perp$ under super-operator
\begin{equation*}
    \F = \frac{1}{K} (\sum_{\alpha\in Act} \E_\alpha+\sum_{M\in \mem}\sum_{M_i\in M} M_i\cdot M_i^\dag),
\end{equation*}
where $K=|Act|+|\mem|$.
A crucial part of the algorithm is to compute $X_s$ for each $s\in Act^*$.
By definition, we have $\E_s(V)\subseteq B^\perp$ whenever $V=\supp(\rho)\subseteq X_s$. Therefore,
\begin{align*}
&X_s = \spans\left(\bigcup\{\supp(\rho): P_B\E_s(\rho)P_B=0\}\right)\\ &= \bigvee\{V: \E_{s}(V)\subseteq B^\perp\}\nonumber =\E_s^{-1}(B^\perp)= (\E_s^*(B))^\perp,\label{eq:Xs}
\end{align*}
where $\E^*$ stands for the dual of super-operator $\E$, i.e. $\E^*=\sum A_i^\dag\cdot A_i$ when $\E=\sum A_i\cdot A_i^\dag$.

1. The correctness of the algorithm is essentially based on the proof of Theorem \ref{thm:supOneSchAny}. Here we give a detailed argument. The algorithm returns $s=\epsilon$ at the first two ``return'' statements where $B$ is not invariant or there is an invariant subspace of $\M$ included in $B^\perp$. Otherwise $b$ is initialized as $b>0$, and the algorithm enters the ``while'' loop.
During the loop, $b$ must decrease at least 1. If not, we have found some $s$ such that $b_s>0$, and for any $v\in Act'^*$, it holds that $b_{s\cdot v} = b_s$. By Lemma \ref{lem:XsSub}, we have $X_s = X_{s\cdot v}$ and $\E_s(X_s)\subseteq X_v\subseteq B^\perp$ for all $v$. Therefore, $\E_s(X_s)$ is an invariant subspace of $\M$ included in $B^\perp$, which is a contradiction.
So, $b$ will be $0$ finally and $\s=s^\omega$ is then an optimal scheduler.

2. We note that the algorithm will run the ``while'' loop at most $d$ times and each time it will run the ``for'' loop within the body of the ``while" loop at most $t^d$ times. So the length of $s$ will be at most $d^2$, as it increases at most $d$ in each running of the ``while'' loop. In the ``for'' loop, the complexity mainly comes from computing $\E_w$. It costs at most $O(d^8)$ because the length of $w$ (i.e. the number of matrix multiplications) is at most $O(d^2)$ and each matrix multiplication costs $O(d^6)$. So the complexity of the algorithm is $O(d\cdot t^d\cdot d^8)=O(d^9t^d)$.
\end{proof}
}

\subsection{Proofs of Theorems \ref{thm:AnySchSuperoperator} and \ref{thm:AnySchSuperoperatorSp}}

We first introduce an auxiliary tool.

\begin{defi}
For any sequence $s\in Act^*$, its repetition degree $\rd(s)$ is inductively defined as follows:
\begin{enumerate}
  \item If there does not exist $t\in Act^+$ and $a,b,c \in Act^*$, such that $s= a\cdot t \cdot b \cdot t \cdot c$, then $\rd(s)=0$.
  \item In general, $\rd(s) = \max\{\rd(t)+1: s= a\cdot t \cdot b \cdot t \cdot c, t\in Act^+, a,b,c \in Act^*\}$.
\end{enumerate}
\end{defi}
It is clear that $\rd(s)=0$ for any $s\in \{\epsilon\}\cup Act$. The following lemma provides a way to estimate the repetition degree $\rd(s)$.
\begin{lem}\label{lem:rdLength}
    Let $\M$ be a qMDP with $\mem=\emptyset$ and $B$ an invariant subspace of $\hs$. Assume $|Act|=k$ and $\dim \hs=d$.
    Then for any sequence $s\in Act^+$ and any $x\geq 0$,
    $$|s|\geq L_x \Rightarrow \rd(s)\geq x.$$ Here, $L_x$ is as the same as in Theorem \ref{thm:AnySchSuperoperator}.
\end{lem}
\begin{proof}
    We prove it by induction on $x$. For the case of $x=0$, it is obvious. For $x=1$, assume $s$ is a sequence with length $|s|\geq L_1=k+1$. Since there is only $k$ possible actions, there must be two different integers $p,q\in [1,k+1]$ such that $s_p=s_q$. Then by definition, $\rd(s)\geq 1$.

Now we suppose that for all $x\leq i$ we have $|t|\geq L_x \Rightarrow \rd(t)\geq x$. Assume $|s|\geq L_{i+1}=(K_i+1)L_i$. Then $s$ can be rewritten as $s=v_1\cdots v_{K_i+1}\cdots$, where for $u\in [1,K_i+1]$, $v_u=s_{(u-1)*L_i+1}\cdots s_{u*L_i}$ is a subsequence of length $L_i$. Since there are only $K_i = k^L_i$ different possible sequences of length $L_i$, there must be two different integers $p,q\in [1,K_i+1]$ such that $v_p=v_q$. By induction assumption, we $\rd(v_p)\geq i$. Therefore, $\rd(s)\geq i+1$. This completes the proof.
\end{proof}

Now we can establish a connection between $\rd(s)$ and $\dim X_s$.
\begin{lem}\label{lem:rdXs}
    Let $\M$ be a qMDP with $\mem=\emptyset$ and $B$ an invariant subspace of $\hs$. If for any $s' \in Act^*$ with $0<|s'|\leq L_q$ and $q=\max_{a\in Act} \dim(X_a)$, and for any initial state $\rho$, the scheduler scheduler $\s=s'^\omega$ satisfies $\Pr\nolimits^{\s}(\rho\vDash \Diamond B)=1,$
    then for any sequence $s\in Act^*$ with $|s|\leq L_q$, there exists a non-empty subsequence $v$ of $s = f\cdot v\cdot g$ such that
    $\dim X_v\leq \max\{q-\rd(s),0\}.$
\end{lem}
\begin{proof}
    We prove it by induction on $\rd(s)$.

    (1) For $s$ with $\rd(s)=0$ and $0<|s|\leq L_q$, we have $X_s\subseteq X_{s_1}$ by Lemma \ref{lem:XsSub}. So, $\dim X_s\leq \dim X_{s_1}\leq q$.

    (2) Suppose for any $s'\in Act^+$ with $\rd(s')=i$ and $|s'|\leq L_q$, there exists a non-empty subsequence $v$ of $s'$, such that $\dim X_v\leq \max\{q-i,0\}$.
Now assume $s$ is a sequence with $\rd(s)=i+1$ and $|s|\leq L_q$.
If $\dim X_s=0$, the claim is true.
    Otherwise, by definition, there exists a non-empty subsequence $t$ of $s$ such that $s=a\cdot t \cdot b \cdot t \cdot c$ and $\rd(t)=i$. By the induction assumption, there exists a non-empty subsequence $u$ of $t=f\cdot u \cdot g$ such that $\dim X_u\leq q-i$. Here $\dim X_u>0$, since $\dim X_s>0$.
Therefore, $s$ can be rewritten as $s = a\cdot f\cdot u \cdot g \cdot b\cdot f\cdot u \cdot g\cdot c$. Let $f' = a\cdot f$, $v=u \cdot g \cdot b\cdot f\cdot u$ and $g'=g\cdot c$. Now we prove $\dim X_v\leq q-i-1$. Since $\dim X_u\leq q-i$ and $X_v\subseteq X_u\neq\emptyset$, we only need to prove $X_v\subsetneq X_u$. We do this by refutation. Suppose $X_v=X_u$. Then by Lemma \ref{lem:XsSub}, we have $X_u=X_{u \cdot g \cdot b\cdot f}$ and
    $\E_{u \cdot g \cdot b\cdot f}(X_u)=\E_{u \cdot g \cdot b\cdot f}(X_{u \cdot g \cdot b\cdot f})\subseteq X_u.$
    Thus, $X_u$ is an invariant subspace under super-operator $\E_{u \cdot g \cdot b\cdot f}$. As $X_u\perp B$, by definition, we have
    $\Pr\nolimits^\s(\rho_0\vDash B)=0$
    for $\s=(u \cdot g \cdot b\cdot f)^\omega$ and $\rho_0=I_{X_u}/\dim X_u$.
    Since $|s|\leq L_q$, we have $|u \cdot g \cdot b\cdot f|\leq|s|\leq L_q$.
    This is a contradiction! Therefore, it must be that $X_v\subsetneq X_u$, and we complete the proof.
\end{proof}

Now we can prove Theorems \ref{thm:AnySchSuperoperator} and \ref{thm:AnySchSuperoperatorSp}.

{\renewcommand{\proofname}{Proof of Theorem \ref{thm:AnySchSuperoperator}}
\begin{proof}
    We only need to prove the ``if'' part because the \textquotedblleft only if\textquotedblright\ is obvious. Assume that $\Pr\nolimits^\s (\rho\vDash \Diamond B)=1$ holds for any initial state and any scheduler $\s=s^\omega$ with $|s|\leq L_q$, where $q=\max_{a\in Act} \dim(X_a)$.
    By Lemma \ref{lem:rdLength}, we have $\rd(s)\geq q$ for all $s$ with $|s|=L_q$. Furthermore, by Lemma \ref{lem:rdXs} and the assumption, we have $\dim X_s\leq\max\{q-\rd(s),0\}=0$ for any sequence $s$ with $|s|=L_q$.
    Thus $\tr(P_B\E_s(\rho))>0$ for any $\rho$.
    Since $\tr(P_B\E_s(\rho))=\tr(\E_s^*(P_B)\rho)$ and $\E_s^*(P_B)=U_sD_sU_s^\dag$ where $D_s=diag\{\lambda_{s,1},\cdots,\lambda_{s,d}\}$, we have $\lambda_{s,i}>0$ for any $i$. Then $\tr(P_B\E_s(\rho))\geq m_s>0$ for any trace-1 operator $\rho$, where $m_s=\min \lambda_{s,i}$.
Consequently, for any scheduler $\s$, it holds that
    $$\Pr\nolimits^\s (\rho\vDash \Diamond B) \geq 1-\lim_{t\ra \infty} (1-m)^t=1,$$
    where $m=\min_{|s|=L_q} m_s >0$. This completes the proof by $q\leq d=\dim \hs$.\end{proof}}

{\renewcommand{\proofname}{Proof of Theorem \ref{thm:AnySchSuperoperatorSp}}
\begin{proof}
    This proof is similar to the proof of Theorem \ref{thm:supOneSp}. We can construct a classical MDP with $S=\{s_x:x\in Y\}$ and check whether $\Pr^\s(s_x\vDash\Diamond B)=1$ for all $s_x$ by noting the following two simple facts:
    \begin{itemize}
      \item for any initial state $\rho$ and any scheduler $\s$, the support of the resulting state after first action/measurement will be in $Y$;
      \item for any $s_x\in S$, we can construct an initial state $\rho=P_x/\tr(P_x)$.
    \end{itemize}
\end{proof}}

\subsection{Proof of Theorem \ref{lem:toRadius}}

Let $\M$ be a qMDP with state Hilbert space $\hs$ and $B$ an invariant subspace of $\M$. For each $\alpha\in Act$, we define a new super-operator: $\mathcal{F}_\alpha(\cdot) = P_{T}\E_\alpha(\cdot)P_{T}$ from $\E_\alpha$, where $T=B^\perp$  is the ortho-complement of $B$ in $\hs$ and $P_T$ is the projection operator onto $T$. Furthermore, let $M_\alpha$ be the matrix representation of $\mathcal{F}_\alpha$.

\begin{lem}\label{lem:toDLI}
    Let $\M$ be a qMDP with $\mem=\emptyset$ and $B$ an invariant subspace of $\M$. Then:
    \begin{enumerate} \item The following two statements are equivalent:\begin{enumerate}\item There exists a scheduler $\s$ such that $\Pr^\s(\rho\vDash\Diamond B)=1$ for all initial states $\rho$. \item There exists $\alpha_1 \alpha_2\cdots\in Act^\omega$ such that $\lim_{n\ra \infty}M_{\alpha_n}\cdots$ $M_{\alpha_1}=0.$\end{enumerate}
      \item The following two statements are equivalent: \begin{enumerate} \item For any scheduler $\s$ and any initial state $\rho$, it holds that $\Pr^\s(\rho\vDash\Diamond B)=1$. \item For any $\alpha_1 \alpha_2\cdots\in Act^\omega$, it holds that $\lim_{n\ra \infty}M_{\alpha_n}\cdots$ $M_{\alpha_1}=0.$\end{enumerate}
    \end{enumerate}
\end{lem}

\begin{proof} 1. It is obvious that (b) $\Rightarrow$ (a) because $\tr(\rho(n,\s))=1$ and the probability in $T$ goes to 0.
We now prove (a) $\Rightarrow$ (b). Suppose that $\s$ is a scheduler required in (a).
    Let $T=\spans\{|1\>,\cdots,|k\>\}$ and $B=\spans\{|k+1\>,\cdots,|d\>\}$. As $\mem=\emptyset$, $\s$ is a sequence of actions, i.e. $\s = s_1 s_2\cdots$ with $s_i\in Act$ for all $i$.
    Since
    \begin{align*}
        \E_\alpha(\rho) &=\sum E_{\alpha,i}\rho E_{\alpha,i}^\dag \\
        &=\sum \begin{pmatrix}
            a_{\alpha,i} & 0\\
            c_{\alpha,i} & b_{\alpha,i}
        \end{pmatrix}
        \begin{pmatrix}
            \rho_T & *\\
            * & \rho_B
        \end{pmatrix}
        \begin{pmatrix}
            a_{\alpha,i}^\dag & c_{\alpha,i}^\dag\\
            0 & b_{\alpha,i}^\dag
        \end{pmatrix}\\
        &=\sum\begin{pmatrix}
            a_{\alpha,i} \rho_T a_{\alpha,i}^\dag & *\\
            * & *
        \end{pmatrix}
    \end{align*}
    and
    \begin{align*}
        \F_\alpha(\rho) &= P_T\left(\sum \begin{pmatrix}
            a_{\alpha,i} & 0\\
            c_{\alpha,i} & b_{\alpha,i}
        \end{pmatrix}
        \rho
        \begin{pmatrix}
            a_{\alpha,i}^\dag & c_{\alpha,i}^\dag\\
            0 & b_{\alpha,i}^\dag
        \end{pmatrix}\right)P_T\\
        &= \sum \begin{pmatrix}
            a_{\alpha,i} & 0\\
            0 & 0
        \end{pmatrix}
        \rho
        \begin{pmatrix}
            a_{\alpha,i}^\dag & 0\\
            0 & 0
        \end{pmatrix},
    \end{align*}
    we have
   $\sigma(n,\s)\equiv P_T\rho(n,\s)P_T = \F_{s_n}\cdots\F_{s_1}(\rho).$
    Moreover, as $\Pr^\s(\rho\vDash\Diamond B)=1$, we have $\lim_{n\ra\infty}\tr(\sigma(n,\s))=0$. As $\sigma(n,\s)$ is a density operator, it follows that $\lim_{n\ra\infty}\sigma(n,\s)=0$.

    Let $\mathcal{G}_n(\cdot)\stackrel{\triangle}{=}\F_{s_n}\cdots\F_{s_1}(\cdot)$. Since $\mathcal{G}_n(\cdot)$ is completely positive, we have $\mathcal{G}_n(\rho)\leq\mathcal{G}_n(I)$ as $I\geq \rho$ for any density operator $\rho$. If we use the matrix norm $$\|A\|=\sup_{\|x\|_2=1}\|Ax\|_2=\sqrt{\lambda_{\max}(A^\dag A)},$$ then it holds that $\|\rho\|=\lambda_{\max} (\rho)\leq \|\sigma\|$ when $\rho\leq\sigma$. As a consequence, we obtain
    $$\|\mathcal{G}_n(I/d)\|<\frac{\epsilon}{4d}
        \Rightarrow \|\mathcal{G}_n(\rho)\|\leq\|\mathcal{G}_n(I)\|<\frac{\epsilon}{4}.$$
For any matrix $R$, we have $R=a_+-a_-+i(b_+-b_-)$, where $a_+,a_-,b_+,b_-\geq 0$ and $a_+a_-=b_+b_-=0$. Furthermore, $$\|a_+\|\leq\|a_+-a_-\|=\|\frac{R+R^\dag}{2}\|\leq \frac{\|R\|+\|R^\dag\|}{2}=\|R\|.$$ The first inequality is because $a_+$ and $a_-$ are both positive and their supports are orthogonal . Therefore, we have
    \begin{align*}
        &\forall \epsilon>0, \exists N\in\mathds{N}, \forall n>N, \forall R\in\mathbb{M}_n(\mathds{C}), \\
        &\|\mathcal{G}_n(R)\|\leq \|\mathcal{G}_n(a_+)\| + \|\mathcal{G}_n(a_-)\|+\|\mathcal{G}_n(b_+)\|+\|\mathcal{G}_n(b_-)\|\\& <\epsilon \|R\|.
    \end{align*}
    Thus, for the matrix represents $\mathbf{A}_n$ of $\mathcal{G}_n$, it holds that $\lim_{n\ra\infty}\mathbf{A}_n=0$, and we complete the proof of part 1.

    2. We actually proved that for each scheduler $\s$ and its corresponding sequence $\mathbf{A}_1,\mathbf{A}_2,\dots$,
    $$\forall \rho, \Pr\nolimits^{\s}(\rho\vDash\Diamond B)=1\Leftrightarrow \lim_{n\ra\infty}\mathbf{A}_n=0$$
    in the proof of part 1. Hence, the conclusion of part 2 follows immediately.
\end{proof}

With the help of the above lemma, we are now able to prove Theorem \ref{lem:toRadius}.

{\renewcommand{\proofname}{Proof of Theorem \ref{lem:toRadius}}
\begin{proof}
    1. If $\underline{\varrho}(\Sigma_\M)<1$, then by definition, there exists a sequence $\mathbf{A}_1,\mathbf{A}_2,\cdots$ such that $$\lim_{n\ra\infty} \|\mathbf{A}_n\|^{1/n}\leq \underline{\varrho}(\Sigma_\M)+\epsilon<1.$$ This implies $\lim_{n\ra\infty} \mathbf{A}_n=0$.
    Conversely, if there exists $\mathbf{A}_1,\mathbf{A}_2,\cdots$ such that $\lim_{n\ra\infty}$ $\mathbf{A}_n=0$, then we can find $A\in\Sigma^m$ with $\|A\|^{\frac{1}{m}}<1$ for some $m$. Thus, $\underline{\varrho}(\Sigma_\M)\leq \lim_{n\ra\infty} \|A^n\|^{\frac{1}{nm}}<1$.

    2. By Theorem 3.10 in \cite{CM05}, we know that $DLI(\Sigma_\M)$ is AAS if and only if $\bar{\varrho}(\Sigma_\M)<1$. Together with Lemma \ref{lem:toDLI}, it completes the proof.
\end{proof}}

\section{Conclusions}\label{Concl}
In this paper, we introduced the notion of quantum Markov decision process (qMDP). Several examples were presented to illustrate how can qMPD serve as a formal model in the analysis of  nondeterministic and concurrent quantum programs. The (un)decidability and complexity of a series of reachability problems for qMDPs were settled, but several others left unsolved (the exact complexity of Problem \ref{prob:reachPrSup}.1 and the general case of Problem \ref{prob:reachPrAnyRhoAnySch}).

Developing automatic tools for reachability analysis of qMDPs is a research line certainly worth to pursue because these tools can be used in verification and analysis of programs for future quantm computers. Another interesting topic for further studies is applications of qMDPs in developing machine learning techniques for quantum physics and control theory of quantum systems.

\bibliographystyle{abbrvnat}

\end{document}